\documentclass[pdflatex,sn-aps]{sn-jnl}

\usepackage{graphicx}%
\usepackage{multirow}%
\usepackage{amsmath,amssymb,amsfonts}%
\usepackage{amsthm}%
\usepackage{mathrsfs}%
\usepackage[title]{appendix}%
\usepackage{xcolor}%
\usepackage{textcomp}%
\usepackage{manyfoot}%
\usepackage{booktabs}%
\usepackage{algorithm}%
\usepackage{algorithmicx}%
\usepackage{algpseudocode}%
\usepackage{listings}%
\usepackage{enumitem}
\usepackage{mathabx}
\usepackage{comment} 
\usepackage{natbib}

\newtheorem{Theorem}{Theorem}[section]
\newtheorem{Definition}[Theorem]{Definition}
\newtheorem{Lemma}[Theorem]{Lemma}
\newtheorem{Cor}[Theorem]{Corollary}

\newtheorem{Assumption}[Theorem]{Assumption}

\renewcommand{\Vec}[1]{\mbox{\boldmath $#1$}}
\renewcommand{\proofname}{\textbf{Proof.}}

\newcommand{\Prec}{\mbox{\rm Prc}}
\newcommand{\Merg}{\delta_v}

\newcommand{\MergP}{\delta_{v'}}
\newcommand{\MergH}{\delta_{\frac{1}{2}}}
\newcommand{\vHalf}{v_{\frac{1}{2}}}
\newcommand{\Exp}{\mathbb E}
\newcommand{\Prob}{\mbox{\rm Pr}}
\newcommand{\BG}{\mbox{\rm BG}}
\newcommand{\ReL}{\mbox{\rm ReL}}
\newcommand{\WVG}{\mbox{\rm WVG}}
\newcommand{\PrmL}[2]{\Pi_{#1}^{* #2}}

\algnewcommand\Input{\item[\textbf{Input:}]} 
\algnewcommand\Output{\item[\textbf{Output:}]}

\theoremstyle{thmstyleone}%
\theoremstyle{thmstyletwo}%

\theoremstyle{thmstylethree}%

\begin{document}

\title[Article Titl]{On Computing the Shapley Value \\ in Bankruptcy Games\\ \large— Illustrated by Rectified Linear Function Games —}


\author*[1]{\fnm{Shunta} \sur{Yamazaki}}\email{yamazaki.s.ap@m.titech.ac.jp}

\author[1]{\fnm{Tomomi} \sur{Matsui}}\email{matsui.t.af@m.titech.ac.jp}

\affil[1]{\orgdiv{Department of Industrial Engineering and Economics}, \\
\orgname{Institute of Science Tokyo}, 
\orgaddress{\street{Meguro-ku}, \city{Tokyo}, 
\postcode{152-8552}, 
\country{Japan}}}


\abstract{
In this research, 
    we discuss a problem 
    of calculating the Shapley value 
    in bankruptcy games.
We show that 
    the decision problem 
    of computing the Shapley value 
    in bankruptcy games is NP-complete.
We also investigate 
    the relationship 
    between the Shapley value of bankruptcy games 
    and the Shapley–Shubik index 
    in weighted voting games.
The relation naturally implies 
    a dynamic programming technique 
    for calculating the Shapley value.  
We also present two recursive algorithms 
    for computing the Shapley value: 
    the first is the recursive completion method 
    originally proposed by O'Neill, 
    and the second is our novel contribution 
    based on the dual game formulation. 
These recursive approaches offer conceptual clarity 
    and computational efficiency, 
    especially when combined with memoisation technique.    
Finally, we propose 
    a Fully Polynomial-Time Randomized Approximation Scheme 
    (FPRAS) based on Monte Carlo sampling, 
    providing an efficient approximation method 
    for large-scale instances.
}



\maketitle

\vspace{-5ex}

\begin{center}
\large
\date{\today}
\end{center}

\medskip

\section{Introduction}

A bankruptcy problem is a situation 
    in which the estate of a debtor must be divided 
    among several claimants, 
    while the estate is insufficient to satisfy all claims. 
The central question is how to allocate 
    the estate fairly among the claimants.
A typical example occurs 
    when a firm goes bankrupt and its liquidation value 
    must be shared among creditors, 
    each holding a specific claim. 

Research on bankruptcy problems has been connected 
    with cooperative game theory.
This connection was formally established 
    in pioneering studies 
    by O'Neill~\cite{o1982problem} 
    and by Aumann and Maschler~\cite{aumann1985game}.
The cooperative games arising from bankruptcy problems 
    are commonly referred 
    to as \textbf{bankruptcy games}. 
For more details about bankruptcy games,
    see~\cite{driessen1988cooperative,thomson2015axiomatic}.

Several applications of bankruptcy games 
    to real-world problems have been studied.
For instance, 
    situations resembling bankruptcy arise 
    in the case of corporate insolvency, 
    such as the bankruptcy 
    of Pacific Gas and Electric Company (PG\&E), 
    where the division of assets among creditors 
    can be studied 
    as a bankruptcy problem~\cite{borm2005constrained}.
Similar ideas appear 
    in the allocation of scarce natural resources, 
    such as water distribution under scarcity, 
    where bankruptcy rules provide equitable 
    and efficient 
    allocation principles~\cite{zheng2022water}.
Another example can be found 
    in the field of education management, 
    where extended bankruptcy models 
    have been applied to the allocation 
    of university budgets across faculties, 
    balancing subjective claims 
    with objective 
    entitlements~\cite{pulido2002game}.

In this paper, we focus 
    on the computational aspects 
    of the Shapley value 
    in bankruptcy games 
    and make the following contributions:
\begin{enumerate}
\item 
    We show that a decision problem of computing the Shapley value in bankruptcy games is NP-complete (Section~\ref{section3}).
\item 
    We investigate the relationship 
    between the Shapley value in bankruptcy games 
    and the Shapley–Shubik index in weighted voting games. 
    This connection naturally leads to an algorithm for computing the Shapley value, 
    based on the dynamic programming technique 
    used for evaluating 
    the Shapley–Shubik index. 
    Our result yields a refinement in the time complexity 
    of a known algorithm.~(see Section~\ref{section4}).
\item 
    We present two recursive algorithms 
    for computing the Shapley value in bankruptcy games.
    We first describe the recursive completion method, 
    originally proposed by O'Neill, 
    which builds values from smaller games obtained by restriction.
    We also discuss the time complexity of the method.
    The second is a novel algorithm proposed by the authors,
    based on the dual game formulation.  (Section~\ref{section:recursive}).
\item 
    We propose 
    a Fully Polynomial-Time Randomized Approximation Scheme 
    (FPRAS) based on Monte Carlo sampling, 
    providing an efficient approximation method 
    for large-scale instances. (Section~\ref{section5}).
\end{enumerate}


\section{Preliminaries}\label{section2}

This section describes notations, definitions, and 
 some known properties.

\subsection*{\bf Bankruptcy Game}

\begin{Definition}[Bankruptcy Games]\label{Definition:bankruptcy game}
Let $N=\{1,2,\dots,n\}$ be a set of players (claimants), 
    and a positive rational number $w_i$ denotes 
    the claim of player $i\in N$.
Let an estate be a positive rational number $E$
    which satisfies $0< E < \sum_{i\in N} w_i$.
The bankruptcy game,
    denoted by $\BG [E; w_1, w_2,\ldots ,w_n]$,
    is the cooperative game $(N,v)$ with a characteristic function
\[
v: S \mapsto \max \left\{ 0,\, E - \sum_{i \in N \setminus S} w_i \right\} \quad (\forall S \subseteq N).
\]
\end{Definition}

\noindent
Throughout this paper, 
we denote the total weight 
$\sum_{i \in N} w_i$ by $W.$ 
For each subset $S \subseteq N$, 
    we define $w(S)=\sum_{i\in S} w_i.$

%

The supermodularity of the above characteristic function $v$
    is established in~\cite{aumann1985game,curiel1987bankruptcy}.
Thus, every bankruptcy game is a convex game 
and has a non-empty core including the Shapley value (defined below).

The characteristic function $v$ of a bankruptcy game $(N,v)$ as given in Definition~\ref{Definition:bankruptcy game} can also be expressed in the following form:
\[
    v(S) = \max \left\{ 0,\, E - \sum_{i \in N \setminus S} w_i \right\} 
= \max \left\{0,\, -w_0 + \sum_{i \in S} w_i \right\} 
\mbox{ where }
w_0 := W - E. 
\]

\noindent
We refer to a game $(N,v)$ as a \emph{Rectified Linear Function Game}, denoted by 
$\ReL [w_0; w_1, w_2, \ldots ,w_n]$, 
when its characteristic function 
 is expressed as
\begin{equation}
    v: S \mapsto  \max\left\{0,\,-w_0 + \sum_{i \in S} w_i \right\}  \quad (\forall S \subseteq N),
    \label{def:Rel}
\end{equation}
where the inequalities $0< w_0 < W$ are satisfied.
Given a vector of positive rational numbers 
 $\Vec{w}=(w_1, w_2,\ldots ,w_n)$ and a pair of positive rational numbers $\{E, w_0\}$, 
we say that the bankruptcy game $\BG[E;\Vec{w}]$ is equivalent to 
the rectified linear function game $\ReL[w_0; \Vec{w}]$ if and only if $E+w_0=W.$ 

Rectified linear function games can also be applied to model various real-world situations.
For example, consider a set of players $N$ cooperating to host an event.
Each player $i\in N$ knows their expected revenue $w_i$ from participating in the event.
Additionally, there is a cost (e.g., venue fee) $w_0$ required to host the event.
A coalition $S\subseteq N$ can host the event if and only if
the total revenue $\sum_{i\in S}w_i$ is greater than or equal to $w_0$.
Therefore,
when coalition $S$ hosts the event,
the benefit obtained by the coalition is the total revenue of players in $S$ minus $w_0$,
i.e., $\sum_{i \in S}w_i - w_0 \; (\geq 0)$.
On the other hand, when coalition $S$ cannot host the event 
($\sum_{i\in S}w_i < w_0$),
the benefit obtained by the coalition is equal to zero.
If we describe the above situation as a characteristic function form game,
    we obtain a rectified linear function game 
    in a straightforward manner. 
In this paper, we shall exploit this equivalence and present some proofs and technical arguments in terms of rectified linear function games, since their formulation is more convenient for analysis.

\subsection*{Shapley Value}

The Shapley value, formally defined below, is a payoff vector introduced by Shapley~\cite{shapley1953value}.

\begin{Definition}[Shapley value]\label{shapleyvalue}
A {\em permutation}  $\pi$  of players in $N$ 
 is a bijection 
    $\pi: \{1,2,\ldots ,n\} \rightarrow N$, 
    and we say that 
   $\pi (j)$ denotes a player at the position $j$ in the permutation $\pi$.
Let $\Pi_N$ be the set of all permutations defined on $N$.
 The Shapley value of characteristic function form game $(N,v)$ is a payoff vector    
    $\phi=(\phi_1, \phi_2, \ldots , \phi_n)$ defined by
\[
  \phi_i=\frac{1}{n!} \sum_{\pi \in  \Pi_N} \Merg (\pi, i) 
    \;\; (\forall i \in N)
\]
where $\Merg(\pi, i)=v(\Prec (\pi,i) \cup \{i\})-v(\Prec (\pi,i))$, 
 called the marginal contribution of player $i$ in permutation $\pi$, where
$\Prec (\pi, i)$ is the set of players in $N$ which precede $i$ in the permutation $\pi$. 
\end{Definition}

The following lemma captures a fundamental property of the Shapley value in bankruptcy games.

\begin{Lemma} \label{Theorem:w_0Monotone}
Let $(N,v)$ and $(N,v')$ be a pair of bankruptcy games defined by 
$\BG [E; \Vec{w}]$ and $\BG [E'; \Vec{w}]$ satisfying $E \leq E'$.
Then, the inequality $\Merg (\pi, i) \leq \MergP (\pi, i)$ holds for each 
$(\pi, i) \in \Pi_N \times N$.
The Shapley values  $\phi$ and $\phi'$ defined by 
$(N, v)$ and $(N, v')$ satisfy that 
$\phi_i \leq \phi'_i$ $(\forall i \in N)$.
\end{Lemma}



\begin{proof}
In the following proof,  we assume that $(N,v)$ and $(N,v')$ be a pair of rectified linear function games defined by 
$\ReL [w_0; \Vec{w}]$ and  $\ReL [w'_0; \Vec{w}]$
satisfying $w_0+E=w'_0+E'=W$ and thus $w_0 \geq w'_0$.
Obviously, we have that $\Merg (\pi, i), \MergP (\pi, i) \in [0,w_i]$.
If $w'_0 \leq w(\Prec (\pi, i))$, 
    then the inequality
\[ 
\begin{array}{r@{}l}
\MergP (\pi, i) &= v' (\Prec (\pi, i) \cup \{i\})- v' (\Prec (\pi, i)) \\
    &=(-w'_0 + w(\Prec (\pi, i) \cup \{i\})) -( -w'_0 + w(\Prec (\pi, i) )) 
    = w_i \geq \Merg (\pi, i)
\end{array}
\]
holds.
If $w(\Prec (\pi, i) \cup \{i\}) \leq  w_0$, 
    then we have the following inequality
\[ 
\Merg (\pi, i) = v (\Prec (\pi, i) \cup \{i\})- v(\Prec (\pi, i))
    =0 -0 =0 \leq \MergP (\pi, i).
\]
Lastly, we consider the remained case that
    $w(\Prec (\pi, i)) < w'_0 \leq w_0 < w(\Prec (\pi, i) \cup \{i\})$.
It is easy to see that
\[
\begin{array}{r@{}l}
\MergP (\pi, i) &= v' (\Prec (\pi, i) \cup \{i\})- v'(\Prec (\pi, i))
    = (-w'_0 + w(\Prec (\pi, i) \cup \{i\})) -0 \\
  &\geq   (-w_0 + w(\Prec (\pi, i) \cup \{i\})) -0 
  =  v (\Prec (\pi, i) \cup \{i\})- v(\Prec (\pi, i))
  = \Merg (\pi, i).
\end{array}
\]
\end{proof}

Previous studies have extensively examined various solution concepts in the context of bankruptcy games.
The pioneering work of O'Neill~\cite{o1982problem} introduced 
    the random arrival rule, 
    which is equivalent to the Shapley value in bankruptcy games.
Aziz~\cite{aziz2013computation} showed that
    the problem of computing the Shapley value is \#P-complete. 
The celebrated paper by Aumann and Maschler~\cite{aumann1985game} 
    proved that the solution concept discussed in the Babylonian Talmud 
    coincides with the nucleolus. 
They also described a polynomial time algorithm 
    for computing the nucleolus of the bankruptcy game.
The pre-kernel and pre-nucleolus for bankruptcy games
    are discussed in~\cite{funaki2006prekernel}.
For a detailed account of the properties of solution concepts in bankruptcy games, see~\cite{thomson2015axiomatic,driessen1988cooperative}.
A dual game representation of bankruptcy games is studied in~\cite{driessen1998greedy}.


\section{Hardness Result}
\label{section3}

In this section, we prove that the problem of computing the Shapley value in bankruptcy games is NP-complete.
The next lemma presents a special instance 
    in which the Shapley value admits 
    a closed-form expression.

 \begin{Lemma}\label{Theorem:PropDist}
For any vector of positive rational numbers
 $\Vec{w}=(w_1, w_2, \ldots ,w_n)$,  
the Shapley value 
of $\BG [\frac{W}{2};\Vec{w}]$ 
is equal to 
$(w_1/2, w_2/2, \ldots ,w_n/2),$ 
where $W=\sum_{i\in N}w_i$.
\end{Lemma}



\begin{proof}
In this proof, we discuss the game  $(N,\vHalf)$
    defined by $\ReL [\frac{W}{2}; \Vec{w}]$, 
    which is equal to $\BG[\frac{W}{2}; \Vec{w}].$
We define 
 $\MergH (\pi, i)
     = \vHalf (\Prec (\pi, i) \cup \{i\}) -\vHalf (\Prec(\pi, i))$
$(\forall (\pi, i) \in \Pi_N \times N)$
and
$w_0=\frac{1}{2}W.$
For any permutation $\pi \in \Pi_N$, 
 we denote its reverse permutation by $\pi^\textup{R}$.

In the following, we show that
\[
  \MergH (\pi, i)+\MergH (\pi^\textup{R}, i)=w_i \quad (\forall i \in N).
\]
by considering  the following three cases;

\vskip 5pt

\begin{enumerate}[label=Case~\arabic*:, leftmargin=*]
  \item $w (\Prec (\pi, i) \cup \{i\}) < w_0,$
  \item $w(\Prec (\pi^\textup{R} ,i)\cup \{i\}  ) < w_0,$
  \item $w (\Prec (\pi, i) \cup \{i\}) \geq w_0$ and
        $w(\Prec (\pi^\textup{R} ,i)\cup \{i\}  ) \geq w_0.$
\end{enumerate}

\vskip 5pt
\noindent
\underline{Case~1}: Let consider the case that $w (\Prec (\pi, i) \cup \{i\}) < w_0$. 
Then, the marginal contribution satisfies 
\[
    \MergH (\pi, i) 
    = \vHalf (\Prec (\pi ,i)\cup \{i\}) - \vHalf (\Prec (\pi ,i))
    =0-0=0.
\]
The reverse permutation $\pi^\textup{R}$ satisfies that
\[
\begin{array}{r@{ }l}
 -w_0 + w(\Prec (\pi^\textup{R},i))
 &= -w_0 + w(N \setminus (\Prec (\pi ,i) \cup \{i\}))
 = -w_0 + w(N)-w(\Prec (\pi ,i) \cup \{i\}) \\
 &> -w_0 +W -w_0 = W -2w_0=W-W=0
 \end{array}
\]
and thus, 
\[
\begin{array}{r@{ }l}
 \MergH (\pi^\textup{R}, i) 
    &= \vHalf (\Prec (\pi^\textup{R} ,i)\cup \{i\}) 
    - \vHalf (\Prec (\pi^\textup{R} ,i)) \\
    &=-w_0 + w (\Prec(\pi^\textup{R} ,i)\cup \{i\} )
    -(-w_0 + w(\Prec(\pi^\textup{R} ,i))) \\
    &=w(\Prec(\pi^\textup{R} ,i))+w_i - w(\Prec(\pi^\textup{R} ,i))=w_i.
\end{array}
\]
We obtain that  $\MergH (\pi, i)+\MergH (\pi^\textup{R}, i)=w_i$.

\vskip 5pt
\noindent
\underline{Case~2}: When $w(\Prec (\pi^\textup{R} ,i)\cup \{i\}  ) < w_0$, 
 we can prove the desired equality similarly to Case 1 by swapping $\pi$ and 
 $\pi^\textup{R}$.

\vskip 5pt
\noindent
\underline{Case 3}: 
Consider the remained case that both
$w (\Prec (\pi, i) \cup \{i\}) \geq w_0$ and
 $w(\Prec (\pi^\textup{R} ,i)\cup \{i\}  ) \geq w_0$ hold. 
The assumption directly implies 
\[
\begin{array}{r@{ }l}
\vHalf (\Prec (\pi, i) \cup \{i\})
=-w_0 +w(\Prec (\pi, i) \cup \{i\}) 
= -w_0 +w(\Prec (\pi, i))+w_i 
\end{array}
\]
and
\[
\begin{array}{r@{ }l}
\vHalf (\Prec (\pi^\textup{R}, i) \cup \{i\})
&=-w_0 +w(\Prec (\pi^\textup{R}, i) \cup \{i\}) 
= -w_0 +w(N \setminus \Prec (\pi, i)) \\
&= -w_0 +W - w(\Prec (\pi, i)).
\end{array}
\]
We also obtain that 
\[
\begin{array}{r@{ }l}
w(\Prec (\pi, i) )
=w(N\setminus ( \Prec(\pi^\textup{R}, i) \cup \{i\})
=W - w( \Prec(\pi^\textup{R}, i) \cup \{i\})\leq W-w_0=W-W/2=w_0
\end{array}
\]
and thus 
\[
\vHalf (\Prec (\pi, i) )= \max \{0, -w_0+w(\Prec (\pi, i))=0.
\]
Similarly, we have that 
$\vHalf (\Prec (\pi^\textup{R}, i) )=0.$

The above properties result in the following 
\[
\begin{array}{r@{ }l}
\MergH (\pi, i)+\MergH (\pi^\textup{R}, i)
&= \vHalf (\Prec (\pi, i) \cup \{i\}) - w(\Prec (\pi, i) )
+\vHalf (\Prec (\pi^\textup{R}, i) \cup \{i\}) - \vHalf (\Prec (\pi^\textup{R}, i) ) \\
&=(-w_0 +w(\Prec (\pi, i))+w_i) -0
 +(-w_0 +W - w(\Prec (\pi, i)))-0 \\
 &=-2w_0+W +w_i =w_i.
\end{array}
\]

The Shapley value $\phi_i$ of player $i$ satisfies
\begin{align*}
\phi_i &= \frac{1}{n!}\sum_{\pi \in \Pi_N} \MergH (\pi, i) 
=\frac{1}{2\cdot n!}\sum_{\pi \in \Pi_N}(\MergH (\pi, i) + \MergH (\pi^{\textup{R}}, i) )
=\frac{1}{2\cdot n!}\sum_{\pi \in \Pi_N}w_i 
 =\frac{1}{2}w_i.
\end{align*}
\end{proof}


The above proof relies on the following fundamental property.

\begin{Cor} \label{Theorem:CorMerg}
For any $(\pi, i)\in \Pi_N \times N$, 
the equality 
$\MergH (\pi, i) + \MergH (\pi^{\textup{R}}, i)=w_i$
holds.
\end{Cor}


We next consider the problem SHP, 
    which concerns the analysis of hardness.

\[
\fbox{\parbox{0.9\textwidth}{
\textbf{Problem SHP} \\
\textbf{Instance:} 
A positive integer vector $\Vec{w}=(w_1, \dots, w_n)$ 
    satisfying that
    $w_i$ is an even integer for each $i \in \{ 1, \dots, n-1\}$, and 
$w_n = 1.$  \\
\textbf{Question:} 
Let $(N,v)$ be a characteristic function form game defined by
$\BG [\lfloor \frac{W}{2} \rfloor; \Vec{w}]$
(which is equivalent to 
$\ReL [\left\lceil \frac{W}{2} \right\rceil ; \Vec{w}]$ ),
where $W=\sum_{i\in N} w_i$. 
Does the Shapley value 
$\phi_n = \frac{1}{n!}
    \sum_{\pi \in \Pi_N} \Merg (\pi ,n)$ 
satisfies $\phi_n < \frac{1}{2}$?
}}
\]

\begin{Theorem}\label{Theorem:NPC}
Problem SHP is NP-complete.
\end{Theorem}

\begin{proof}
(1) Problem SHP belongs to the class NP.

We show that $\phi_n < \frac{1}{2}$ if and only if 
  $\exists \pi \in \Pi_N$ satisfying $\Merg (\pi, n) + \Merg (\pi^{\textup{R}}, n) <1$.
Lemma~\ref{Theorem:w_0Monotone} and Corollary~\ref{Theorem:CorMerg}  imply that 
\[
  \Merg (\pi, n)+\Merg (\pi^{\textup{R}}, n) 
  \leq \MergH (\pi, n) + \MergH(\pi^{\textup{R}}, n) = w_n=1.
\] 
Thus, we have that 
\[
\phi_n = \frac{1}{n!} \sum_{\pi \in \Pi_N} \Merg (\pi ,n)
    =\frac{1}{2\cdot n!} \sum_{\pi \in \Pi_N} 
        (\Merg (\pi ,n)+ \Merg(\pi^{\textup{R}}, n))
        \leq  \frac{1}{2\cdot n!} \sum_{\pi \in \Pi_N} 1  =\frac{1}{2}.
\]
The above properties directly imply that 
    $\phi_n =(1/2)$
    if and only if 
    \mbox{$\Merg (\pi ,n)+ \Merg (\pi^{\textup{R}}, n)=1$}
    holds for any $\pi \in \Pi_N$.
Now we have a desired result and 
    a permutation $\pi$ satisfying 
     \mbox{$\Merg (\pi ,n)+ \Merg (\pi^{\textup{R}}, n)< 1$}  
     is a polynomial size certificate when problem SHP has YES answer.

\smallskip 
\noindent
(2) We show that the known NP-complete problem PARTITION can be reduced to problem SHP in polynomial time.

\[
\fbox{\parbox{0.9\textwidth}{
\textbf{PARTITION Problem}\\
\textbf{Instance:} 
A positive integer vector 
 $(a_1, \dots, a_m)$. \\
\textbf{Question:} 
Does there exist a subset $S \subseteq \{1,2,\ldots ,m\}$ satisfying 
$\displaystyle 
 \sum_{i \in S}a_i 
   = \frac{1}{2}\sum_{i=1}^m a_i$?
}}
\]
We construct an instance of problem SHP from that of PARTITION as follows:
we set  $n = m+1$, 
\[
w_i = \left\{ \begin{array}{ll}
   2a_i & (i \in \{1,2,\dots, n-1\}), \\
   1     & (i=n=m+1),
   \end{array} \right. \mbox{ and }  
w_0 =  \frac{1}{2} \sum_{i=1}^{n-1} w_i + 1  
= \sum\limits_{i=1}^{m} a_i + 1.
\]
Now, we prove that the PARTITION problem instance has YES answer if and only if the answer to problem SHP instance is Yes.

\noindent\underline{
[PARTITION has Yes answer $\Rightarrow$ SHP has Yes answer]
}

As problem PARTITION has  Yes answer, 
 there exists $S \subseteq \{1,2,\ldots ,n-1\}$ 
 satisfying $\sum_{i \in S} a_i =  \frac{1}{2}\sum_{i=1}^{m}a_i$. 
We construct a permutation $\pi_{*} \in \Pi_N$ 
 satisfying $\Prec (\pi_{*}, n)=S$.
Then, it is easy to see that 
\[
 w(\Prec(\pi_{*}, n)\cup \{n\}) 
 =\frac{1}{2}\sum_{i=1}^{m}a_i +1
 = w_0.
 \]
Thus, we have that 
\[
\Merg (\pi_{*}, n) 
=v(\Prec(\pi_{*}, n)\cup \{n\})- v(\Prec(\pi_{*}, n))  ) =0-0=0.
\]
Similarly, the equalities 
\[
 w(\Prec(\pi_{*}^{\textup{R}}, n)\cup \{n\}) 
 =\frac{1}{2}\sum_{i=1}^{m}a_i +1
 = w_0
\]
imply that 
\[
\Merg (\pi_{*}^{\textup{R}}, n) 
=(v(\Prec(\pi_{*}^{\textup{R}}, n)\cup \{n\})- v(\Prec(\pi_{*}^{\textup{R}}, n))  ) 
 =0-0=0.
\]
As shown in (1), the inequality 
\[
\Merg (\pi_{*}, n)+ \Merg (\pi_{*}^{\textup{R}}, n)
    =0+0 =0 <1
\]
implies $\phi_n < \frac{1}{2}$. 

\noindent\underline{
[SHP has Yes answer $\Rightarrow$ PARTITION has Yes answer]
}

In the following, we show the contraposition.
Assume that PARTITION has NO answer, i.e., 
    $\forall S \subseteq \{1,2,\ldots n-1\},$
    $w(S) \neq (1/2)\sum_{i=1}^{n-1} w_i$.
We only need to show that $\forall \pi \in \Pi_N$, 
    $\Merg (\pi, n)+\Merg (\pi^{\textup{R}},n)=1$.
    
For any permutation $\pi \in \Pi_N$, we have the following.
The assumption implies that  $w(\Prec (\pi, n)) \neq (1/2)\sum_{i=1}^{n-1} w_i$.
Consider the case that $w(\Prec (\pi, n)) < (1/2)\sum_{i=1}^{n-1} w_i$.
As $\{w_1, w_2, \ldots ,w_{n-1}\}$ is a set of even integers, 
the inequality  $w(\Prec (\pi, n)) \leq  (1/2)\sum_{i=1}^{n-1} w_i -1$ holds.
Thus, we have that
\[
    w(\Prec (\pi, n)) \leq  w(\Prec (\pi, n)\cup \{n\}) 
    = w(\Prec (\pi, n)) +1 \leq  (1/2)\sum_{i=1}^{n-1} w_i =w_0,
\]
which leads to the equality
\[
    \Merg (\pi, n)=v (\Prec (\pi,n) \cup \{n\})-v(\Prec (\pi, n))=0-0=0.
\]
Similarly, we have that 
\[
\begin{array}{r@{}l}
  w(\Prec (\pi^{\textup{R}}, n) \cup \{n\}) 
  & \geq w(\Prec (\pi^{\textup{R}}, n))   
  = (W-w(\Prec(\pi ,n) \cup \{n\})) \\
  &\displaystyle 
    =\sum_{i=1}^{n-1}w_i+1 -( w(\Prec(\pi ,n)) +w_n) \\
  &\displaystyle 
    \geq  \sum_{i=1}^{n-1}w_i+1- ((1/2) \sum_{i=1}^{n-1}w_i -1) -w_n
 = (1/2) \sum_{i=1}^{n-1}w_i +1 = w_0,
\end{array}
\]
which leads to the equality
\[
\begin{array}{r@{}l}
 \Merg (\pi^{\textup{R}}, n)
 &=v (\Prec (\pi^{\textup{R}},n) \cup \{n\})-v(\Prec (\pi^{\textup{R}}, n)) \\
 &=(-w_0+ w(\Prec (\pi^{\textup{R}},n)\cup \{n\})) 
 -(-w_0+ w(\Prec (\pi^{\textup{R}},n))) \\
 &= w(\Prec (\pi^{\textup{R}},n))+w_n - w(\Prec (\pi^{\textup{R}},n))=1.
\end{array} 
\]
From the above, the equality   $\Merg (\pi, n)+\Merg (\pi^{\textup{R}},n)=1$ holds.
In a similar way, we can deal with the remaining case that 
    $w(\Prec (\pi, n)) > (1/2)\sum_{i=1}^{n-1} w_i$.
\end{proof}


\section{
Relation to Weighted Voting Games}\label{section4}


Lemma~\ref{Theorem:w_0Monotone} demonstrates that 
    the Shapley value is non-decreasing 
    in the estate $E$.
In this section, we present a detailed parametric analysis
    of the Shapley value as a function of the estate $E$.
Our result reveals a connection between 
    the Shapley value of bankruptcy games
    and the Shapley-Shubik index in weighted voting games.
Throughout this section, we assume that
 $E,w_0,w_1,\ldots ,w_n$ are positive integers.

\begin{Definition}[Weighted Voting Game]\label{Definition:weighted voting game}
Let $N = \{1,2,\ldots,n\}$ be a set of players,
   $\Vec{w}=(w_1, w_2, \ldots, w_n)$ be a vector of positive integers,
   and $q$ be a positive integer satisfying 
   $0 < q \leq w_1 + \cdots + w_n$.
Here, $w_i$ denotes the number of votes held by player $i$, 
    and $q$ denotes the quota, 
    i.e., the minimum number of votes required for a coalition to win. 
A weighted voting game,
    denoted by $\WVG[q; \Vec{w}]$,
    is a characteristic function form game 
    $(N, v^{\langle q \rangle})$, 
    whose characteristic function is defined as
\[
v^{\langle q \rangle} (S) =
\begin{cases}
1 & (\mbox{if } S \neq \emptyset \mbox{ and } \sum_{i \in S} w_i \geq q), \\
0 & (\mbox{otherwise}).
\end{cases}
\]
\end{Definition}

The Shapley value of a weighted voting game 
    is called the 
    \emph{Shapley–Shubik index} (see~\cite{ShapleyShubik1954}).
Given a quota $q$,  
    define $\varphi^{\langle q \rangle}$ as 
    the Shapley–Shubik index
    (Shapley value) of $\WVG [q; \Vec{w}].$ 
For a detailed account of algorithms for calculating power indices
    in weighted voting games, 
    see~\cite{Mann1962RM3158,matsui2000survey,uno2012efficient,chalkiadakis2011computational}.

\begin{Theorem}\label{Theorem:ssindex}
Assume that
 $E,w_1,\ldots ,w_n$ are positive integers
 satisfyig $E<\sum_{i \in N} w_i$.
The Shapley value $\phi$ of $\BG[E; \Vec{w}]$  
    satisfies
$
 \phi= \sum_{q = 1}^{E} \varphi^{\langle q \rangle}, 
$
    where $\varphi^{\langle q \rangle}$ is the 
    Shapley–Shubik index
    of $\WVG[q; \Vec{w}]$.
\end{Theorem}

\begin{proof}
Let $(N,v)$ be a game defined by $\BG [E; \Vec{w}]$.
The dual game of $(N,v)$ is the game 
    $(N,v^*)$ defined by 
    $v^*(S)=v(N)-v(N \setminus S)$ 
    for every $S \subseteq N$.
It is obvious that 
\[
    v^*(S)
    =E-\max \left\{ 
        0, E-\sum_{i \in N\setminus (N \setminus S)} w_i 
    \right\}
    =\min \left\{
    E, \sum_{i\in S} w_i
    \right\}
    =\sum_{q=1}^{E} v^{\langle q \rangle} (S)
    \;\;\; (\forall S \subseteq N).
\]
The additivity (law of aggregation~\cite{shapley1953value})
    of the Shapley value implies that
    $\phi^*=\sum_{q=1}^E \varphi ^{\langle q \rangle}$
    where $\phi^*$ is the Shapley value of $(N,v^*)$.
The self-duality of the Shapley value
    (see, e.g.,~\cite{peleg2007introduction})
    implies $\phi^*=\phi$ 
    and thus we obtain the desired result.
\end{proof}

The above theorem implies that 
    the marginal contribution 
    to the Shapley value 
    resulting from an increase 
    in the estate from $E-1$ to $E$
    coincides with the Shapley–Shubik index 
    of the weighted voting game $\WVG [E;\Vec{w}]$.

By employing Theorem~\ref{Theorem:ssindex},
    one can compute the Shapley value of a bankruptcy game 
    by executing an algorithm for computing 
    the Shapley–Shubik index of a weighted voting game 
    $E$ times.
Furthermore, 
    by modifying the algorithm for computing 
    the Shapley–Shubik index, 
    it may be possible to construct a more efficient method than simply running it $E$ times.

In the following, we describe an algorithm
    based on the dynamic programming technique,
    whose time complexity is bounded by $O(n^2 E)$ for each player. 
The definition of weighted voting game 
    $(N,v^{\langle q\rangle})$ implies that
\[
    \begin{array}{rcl}
    \varphi_i^{\langle q \rangle}&=&\displaystyle 
    \sum_{S\subseteq N \setminus \{i\}}
    \frac{|S|! (n-|S|-1)!}{n!}
        (v^{\langle q\rangle}(S \cup \{i\})
        -v^{\langle q\rangle}(S)) \\
    &=& \displaystyle 
    \sum_{t=0}^{n-1} \sum_{w=q-w_i}^{q-1}
    \frac{t! (n-t-1)!}{n!} c_i (w,t),
   \end{array} 
\]
where 
\[
    c_i (w,t)
    =\left\{ \begin{array}{ll} 
    \displaystyle 
    \# \left\{S \subseteq N\setminus \{i\} 
        \left|  \sum_{i \in S} w_i =w \mbox{ and } |S|=t 
        \right. \right\} 
        & (\mbox{ if } 1\leq w\leq W \mbox{ and } 1\leq t \leq n-1), \\
    1 & (\mbox{ if } w=t=0), \\
    0 & (\mbox{otherwise}).
    \end{array} \right.
\]

Theorem~\ref{Theorem:ssindex} implies that
\[
\begin{array}{rcl}
\phi_i 
&=& \displaystyle 
 \sum_{q =1}^{E} \sum_{t=0}^{n-1}  \sum_{w=q-w_i}^{q-1}
    \frac{t! (n-t-1)!}{n!} c_i (w,t) \\
&=& \displaystyle 
 \sum_{t=0}^{n-1} 
 \left( 
  \frac{t! (n-t-1)!}{n!}  
    \sum_{w = 0}^{E-1} 
    \left(
        \min\{ E-w, w_i\} 
        c_i (w,t)
    \right)
 \right),  
 \end{array}
\]  
which allows us to compute the Shapley value $\phi_i$ 
in $O(nE)$ time using
the matrix $C_i=(c_i(w,t))$
defined by components satisfying 
$0\leq w \leq E-1$ and $0\leq t \leq n-1$.
An ordinary dynamic programming technique gives 
    an algorithm for calculating the matrix $C_i$
    in $O(n^2 E)$ time
    (see, e.g.,~\cite{matsui2000survey,chalkiadakis2011computational}).
Thus, there exists an algorithm for calculating 
    the Shapley value of a bankruptcy game 
    whose time complexity is bounded by $O(n^2 E)$ for each player. 

According to Theorem~\ref{Theorem:ssindex}, one can easily construct an algorithm for computing the Shapley–Shubik index by applying either of the recursive algorithms (Algorithm~\ref{algo-recL} or Algorithm~\ref{algo-recD} in the next section) to the games with estates $q$ and $q-1$ respectively, and taking the difference between the resulting Shapley values.
\section{Recursive Algorithms}\label{section:recursive}

In this section, 
    we describe recursive algorithms.
Throughout this section, we assume that 
    $E,w_0, w_1, \ldots ,w_n$ 
    are positive rational numbers 
    satisfying $E+w_0= W=\sum_{i \in N} w_i.$
 For each pair $(S',j)$ satisfying 
    $j\in S' \subseteq N$, 
    we define 
    $\PrmL{S'}{j}$ 
        to be the set of permutations on $S'$ 
        whose last component is player $j$.

\subsection{Recursive Completion Method}

In the following, 
    we describe a recursive completion method
    proposed in~\cite{o1982problem}, 
    which is based on the following property.

\begin{Theorem} \label{Theorem:recursiveL}
Let $(N,v)$  be a characteristic function form game 
    defined by $\ReL [w_0;\Vec{w}]$, 
    and let  $\phi$ denote
    its Shapley value.
Given a subset $S' \subseteq N$
    with $|S'|\geq 2$ and 
    $\sum_{k\in S'}w_k > w_0,$
    the restricted game $(S', v^{[S']})$
    is defined by 
    $v^{[S']}(S)=\max \{0, -w_0+\sum_{k\in S}w_k\}$
     $(\forall S \subseteq S').$
For every non-empty subset $S' \subseteq N,$    
    we define a vector $\phi^{[S']}$ 
    indexed by $S'$ as follows:
\[
    \phi^{[S']} =
    \begin{cases}
    \mbox{the Shapley value of } 
        (S', v^{[S']}) 
        &(\mbox{if } |S'|\geq 2 \mbox{ and } \sum_{k\in S'}w_k>w_0),\\
    \Vec{0} & (\mbox{if } |S'|\geq 2 \mbox{ and } \sum_{k\in S'}w_k\leq w_0), \\
    v(\{k\}) &  (\mbox{if } |S'|= 1 \mbox{ and } S'=\{k\}). \\
    \end{cases}
\]

Then, whenever  $|S'|\geq 2$ and  $\sum_{k\in S'}w_k>w_0$,
    the components of $\phi^{[S']}$ satisfy
    the recursive relation
\[
    \phi^{[S']}_i=
    \frac{1}{|S'|}
    \left(
        \min \{w_i,v(S') \} 
        + \sum_{j\in S'\setminus \{i\}} 
            \phi^{[S'\setminus \{j\}]}_i
    \right)
    \quad (\forall i \in S').
\]
\end{Theorem}


\begin{proof}

    It is obvious that 
    for every subset 
    $S'' \subseteq N$ such that  
    $|S''|\geq 2$ and  $\sum_{k\in S''}w_k>w_0$, 
    we have 
\[
    v^{[S'']} (S)
        =\max \left\{ 0, -w_0 +\sum_{k\in S} w_k \right\}
        =v(S) \;\; (\forall S \subseteq S'').
\]

If $|S'|\geq 2$ and  $\sum_{k\in S'}w_k > w_0,$
    then the Shapley value $\phi^{[S']}$ 
    satisfies that 
\begin{eqnarray*}
 \phi^{[S']}_i
 &=&
 \frac{1}{|S'|!}\sum_{\pi \in \Pi_{S'}} 
    \delta_{v^{[S']}}(\pi,i) 
  = \frac{1}{|S'|!} 
  \left( 
    \sum_{\pi \in \PrmL{S'}{i}} \delta_{v^{[S']}}(\pi,i) 
    + \sum_{j\in S' \setminus \{i\}} 
      \sum_{\pi \in \PrmL{S'}{j}} \delta_{v^{[S']}}(\pi,i) 
  \right) \\
  &=&\frac{1}{|S'|}
  \left( \min \{w_i, v(S')\}   
    + \sum_{j\in S' \setminus \{i\}} 
     \sum_{\pi \in \PrmL{S'}{j}} 
   \frac{ v^{[S']}(\Prec (\pi,i)\cup \{i\})
    - v^{[S']}(\Prec (\pi,i))}{(|S'|-1)!}
    \right) 
    \\
  &=&\frac{1}{|S'|}
  \left( \min \{w_i, v(S')\}   
    + \sum_{j\in S' \setminus \{i\}} 
     \sum_{\pi \in \PrmL{S'}{j}} 
   \frac{ v(\Prec (\pi,i)\cup \{i\})
    - v(\Prec (\pi,i))}{(|S'|-1)!}
    \right). 
    \\
\end{eqnarray*}

\noindent
We now turn to the final term introduced above
    for each $j \in S'\setminus \{i\}.$

\noindent
(1) In case that  $|S' \setminus \{j\}| \geq 2$ and
    $\sum_{k\in S'\setminus\{j\}}w_k>w_0$, we have that  
\begin{eqnarray*}
\lefteqn{
  \sum_{\pi \in \PrmL{S'}{j}} 
   \frac{ v(\Prec (\pi,i)\cup \{i\})
    - v(\Prec (\pi,i))}{(|S'|-1)!}
}\\
&=&  \sum_{\pi \in \PrmL{S'}{j}}
   \frac{ v^{[S'\setminus \{j\}]}(\Prec (\pi,i)\cup \{i\})
    - v^{[S'\setminus \{j\}]}(\Prec (\pi,i))}{(|S'|-1)!}
=  \sum_{\pi \in \PrmL{S'}{j}}
   \frac{ \delta_{v^{[S'\setminus \{j\}]}} (\pi,i)}{(|S'|-1)!}\\
&=& \frac{1}{|S' \setminus \{j\}|!} 
    \sum_{\pi \in \Pi_{S'\setminus \{j\}}} 
        \delta_{v^{[S'\setminus \{j\}]}} (\pi,i)
= \phi^{[S' \setminus \{j\}]}_i. 
\end{eqnarray*}

\noindent
(2) Consider the case that 
    $|S' \setminus \{j\}|\geq 2$ and $\sum_{k\in S'\setminus\{j\}}w_k\leq w_0.$
As $j \neq i$, every permutation $\pi \in \PrmL{S'}{j}$ satisfies that      
    $\Prec (\pi,i) \subseteq
    \Prec (\pi, i) \cup \{i\} \subseteq S' \setminus \{j\}$
    and thus 
    $0\leq v(\Prec (\pi, i))
    \leq v(\Prec (\pi, i)\cup \{i\} ) 
    \leq v(S' \setminus \{j\})=0$.
From the above, we obtain 
\[
  \sum_{\pi \in \PrmL{S'}{j}} 
   \frac{ v(\Prec (\pi,i)\cup \{i\})
    - v(\Prec (\pi,i))}{(|S'|-1)!}
  =\sum_{\pi \in \PrmL{S'}{j}} 
   \frac{ 0-0}{(|S'|-1)!}
   =0=\phi^{[S' \setminus \{j\}]}_i.
\]

\noindent
(3) Finally, we consider the case 
    where $|S'\setminus \{j\}|=1.$
    It implies that $S'=\{i,j\}$ and 
    $\PrmL{S'}{j}$ contains the unique permutation $(\pi(1),\pi(2))=(i,j)$, 
    from which it follows that
\[
  \sum_{\pi \in \PrmL{S'}{j}} 
   \frac{ v(\Prec (\pi,i)\cup \{i\})
    - v(\Prec (\pi,i))}{(|S'|-1)!}
  = \frac{v(\{i\}) -v(\emptyset )}{1!}
   =v(\{i\})
   =\phi^{[\{i\}]}_i=\phi^{[S' \setminus \{j\}]}_i.
\]    
\end{proof}


\smallskip 

The recursive formula 
    in Theorem~\ref{Theorem:recursiveL}, 
    leads to Algorithm~\ref{algo-recL}.
In the implementation of Algorithm~\ref{algo-recL}, 
    the memoisation technique is employed,
    which stores
    previously computed results so as to prevent 
    redundant computations.
The use of this technique allows 
    the time complexity to be bounded 
    by $O(n^2 {\cal L})$, 
    assuming sufficient memory resources, 
    where ${\cal L}
    =\# \{S' \subseteq N  \mid |S'|\geq 2, v(S')>0 \}
    =\# \{S'' \subseteq N \mid 
        |S''| \leq n-2, \sum_{k\in S''} w_k < E\}$.  

%
%

\begin{algorithm}
\caption{Recursive Completion Method by O'Neill~\cite{o1982problem}}\label{algo-recL}
\begin{algorithmic}
\Input{A positive rational vector $(w_0; w_1, w_2, \dots, w_n)$}
    satisfying $w_0 < \sum_{i=1}^n w_i.$
\Output{A payoff vector $\phi = (\phi_1, \phi_2, \dots, \phi_n)$.}
\Function{RC}{$S'$}
    \State $v\_S' \gets \max\{0,-w_0+\sum_{k\in S'}w_k\}$
    \If{$|S'| = 1$}
        \State \Return $v\_S'$
    \ElsIf{$\sum_{k\in S'}w_k\leq w_0$}
        \State \Return \Vec{0} (\mbox{the zero vector indexed by  $S'$})
    \Else
        \For{each $j \in S'$}
            \State $\phi^{[S'\setminus\{j\}]} \gets \Call{RC}{S'\setminus \{j\}}$
        \EndFor
        \For{each $i \in S'$}
            \State $\displaystyle
            \phi^{[S']}_i \gets
            \frac{1}{|S'|}
            \left(
                \min\{w_i, v\_S'\}
                + \sum_{j\in S'\setminus \{i\}}\phi^{[S'\setminus\{j\}]}_i
            \right)$
        \EndFor
        \State \Return $(\phi^{[S']}_i)_{i \in S'}$
    \EndIf
\EndFunction

\Statex

\State Compute $\phi \gets \Call{RC}{N}$
\State \Return $\phi$
\end{algorithmic}
\end{algorithm}

\subsection{
A Recursive Algorithm via the Dual Game
} 

In this subsection,
    we propose a recursive algorithm
    based on the dual game $(N,v^*)$ 
    of $\BG [E; \Vec{w}]$,
    which is defined by 
    $v^*(S)=\min \{E, \sum_{k \in S} w_k\}$ $(\forall S \subseteq N).$

\begin{Theorem} \label{Theorem:recursiveD}
Let $(N,v^*)$ be the dual game of $\BG [E;\Vec{w}]$
    and denote its Shapley value by $\phi^*$.
Given a subset $S' \subseteq N$
    satisfying $|S'|\geq 2$ 
    and $\sum_{k \in S'} w_k >E$, 
    the restricted game 
    $(S', v^{*[S']})$
    is defined by 
    $v^{*[S']}(S)=\min \{E, \sum_{k\in S} w_k\}$
    $(\forall S \subseteq S').$
For every non-empty subset $S' \subseteq N,$   
    we define the vector $\phi^{*[S']}$ 
    indexed by $S'$ as follows:
\[
    \phi^{*[S']}=
    \begin{cases}
    \mbox{the Shapley value of }  (S', v^{*[S']})
    & (\mbox{if } |S'|\geq 2
        \mbox{ and } \sum_{k\in S'} w_k > E
        ), \\
    (w_i)_{i\in S'} 
    & (\mbox{if } |S'|\geq 2
        \mbox{ and } \sum_{k\in S'} w_k \leq E
      ), \\
     v^*(\{k\}) 
     & ( \mbox{if } |S'|= 1 \mbox{ and } S'=\{k\}).  \\
    \end{cases}  
\]    

Then, for any $S'\subseteq N$ 
    satisfying  $|S'|\geq 2$ 
    and  $\sum_{k\in S'} w_k >E$,
    the vector $\phi^{[S']}$ satisfies the recursive relation
    
\[
    \phi^{*[S']}_i=
    \frac{1}{|S'|}
    \left(
        v^*(S')-v^*(S' \setminus \{i\})  
        + \sum_{j\in S'\setminus \{i\}} 
            \phi^{*[S'\setminus \{j\}]}_i
    \right) 
    \qquad (\forall i \in S').
\]
\end{Theorem}

\begin{proof}
    Whenever 
    $S'' \subseteq N$ satisfies   
    $|S''|\geq 2$ and  $\sum_{k\in S''} w_k > E$, 
    we have 
\[
    v^{*[S'']} (S)
        =\min \left\{ E, \sum_{k\in S} w_k \right\}
        =v^*(S) \;\; (\forall S \subseteq S'').
\]

If $|S'|\geq 2$ and $\sum_{k\in S'}w_k > E$, 
    then the Shapley value $\phi^{*[S']}$ satisfies that 
\begin{eqnarray*}
 \phi^{*[S']}_i
 &=&
 \frac{1}{|S'|!}\sum_{\pi \in \Pi_{S'}} 
    \delta_{v^{*[S']}}(\pi,i) 
  = \frac{1}{|S'|!}
  \left( 
    \sum_{\pi \in \PrmL{S'}{i}} \delta_{v^{*[S']}}(\pi,i) 
    + \sum_{j\in S' \setminus \{i\}} 
      \sum_{\pi \in \PrmL{S'}{j}} \delta_{v^{*[S']}}(\pi,i) 
  \right) \\
  &=&\frac{1}{|S'|}
  \left( 
    \begin{array}{l}
    \displaystyle
    v^{*[S']}(S')-v^{*[S']}(S' \setminus \{i\}) \\
    \displaystyle \qquad \qquad
    + \sum_{j\in S' \setminus \{i\}} 
     \sum_{\pi \in \PrmL{S'}{j}} 
   \frac{ v^{*[S']}(\Prec (\pi,i)\cup \{i\})
    - v^{*[S']}(\Prec (\pi,i))}{(|S'|-1)!}
    \end{array}
    \right) 
    \\
  &=&\frac{1}{|S'|}
  \left( 
    \begin{array}{l}
    \displaystyle
    v^*(S')-v^*(S' \setminus \{i\}) \\
    \displaystyle \qquad \qquad
    + \sum_{j\in S' \setminus \{i\}} 
     \sum_{\pi \in \PrmL{S'}{j}} 
   \frac{ v^*(\Prec (\pi,i)\cup \{i\})
    - v^*(\Prec (\pi,i))}{(|S'|-1)!}
    \end{array}
    \right) 
    \\
\end{eqnarray*}

\noindent
We now turn to the final term introduced above
    for each $j \in S'\setminus \{i\}.$

\noindent
(1) In case that  $|S' \setminus \{j\}| \geq 2$ and
    $\sum_{k\in S' \setminus \{j\}} w_k>E$, 
    we have that  
\begin{eqnarray*}
\lefteqn{
  \sum_{\pi \in \PrmL{S'}{j}} 
   \frac{ v^*(\Prec (\pi,i)\cup \{i\})
    - v^*(\Prec (\pi,i))}{(|S'|-1)!}
}\\
&=&  \sum_{\pi \in \PrmL{S'}{j}}
   \frac{ v^{*[S'\setminus \{j\}]}(\Prec (\pi,i)\cup \{i\})
    - v^{*[S'\setminus \{j\}]}(\Prec (\pi,i))}{(|S'|-1)!}
=  \sum_{\pi \in \PrmL{S'}{j}}
   \frac{ \delta_{v^{*[S'\setminus \{j\}]}} (\pi,i)}{(|S'|-1)!}\\
&=& \frac{1}{|S' \setminus \{j\}|!} 
    \sum_{\pi \in \Pi_{S'\setminus \{j\}}} 
        \delta_{v^{*[S'\setminus \{j\}]}} (\pi,i)
= \phi^{*[S' \setminus \{j\}]}_i. 
\end{eqnarray*}

\noindent
(2) Consider the case that 
    $|S' \setminus \{j\}|\geq 2$ and
    $\sum_{k\in S' \setminus \{j\}} w_k \leq E$.
As $j \in S' \setminus \{i\}$, 
    every permutation $\pi \in \PrmL{S'}{j}$ satisfies that      
    $\Prec (\pi,i) \subseteq
    \Prec (\pi, i) \cup \{i\} \subseteq S' \setminus \{j\}$
    and thus 
\[
      \sum_{k\in \Prec (\pi, i)} w_k 
 \leq \sum_{k\in \Prec (\pi, i)\cup \{i\}} w_k  
 \leq \sum_{k\in S' \setminus \{j\}} w_k \leq E.
\]
From the above, we obtain 
\begin{eqnarray*}
\lefteqn{
  \sum_{\pi \in \PrmL{S'}{j}} 
   \frac{ v^*(\Prec (\pi,i)\cup \{i\})
    - v^*(\Prec (\pi,i))}{(|S'|-1)!} 
}\\
 &=&\frac{1}{(|S'|-1)!}
  \sum_{\pi \in \PrmL{S'}{j}} 
  \left(
      \sum_{k\in \Prec (\pi,i)\cup \{i\}} w_k 
    - \sum_{k\in \Prec (\pi,i)} w_k
  \right) 
  =\frac{
    \sum_{\pi \in \PrmL{S'}{j}} w_i
  }{(|S'|-1)!}
  =w_i=\phi^{*[S' \setminus \{j\}]}_i.
\end{eqnarray*}

\noindent
(3) Finally, we consider the case 
    where $|S'\setminus \{j\}|=1.$
    It implies that $S'=\{i,j\}$ and 
    $\PrmL{S'}{j}$ contains the unique permutation $(\pi(1),\pi(2))=(i,j)$, 
    from which it follows that
\[
  \sum_{\pi \in \PrmL{S'}{j}} 
   \frac{ v^*(\Prec (\pi,i)\cup \{i\})
    - v^*(\Prec (\pi,i))}{(|S'|-1)!}
  = \frac{v^*(\{i\}) -v^*(\emptyset )}{1!}
   =v^*(\{i\})
   =\phi^{*[\{i\}]}_i
   =\phi^{*[S' \setminus \{j\}]}_i.
\]    

\end{proof}

The recursive formula 
    in Theorem~\ref{Theorem:recursiveD}, 
    leads to Algorithm~\ref{algo-recD}.
The self-duality of Shapley value implies that
    $\phi=\phi^*$ and thus 
    Algorithm~\ref{algo-recD} computes the Shapley value
    of the given Bankruptcy game.

\begin{algorithm}[h]
\caption{Recursive Algorithm via the Dual Game}\label{algo-recD}
\begin{algorithmic}[Algorithm D]
\Input{A positive rational vector $(E; w_1, w_2, \dots, w_n)$}
    satisfying $E < \sum_{i=1}^n w_i.$
\Output{A payoff vector $\phi^* = (\phi_1^*, \phi_2^*, \dots, \phi_n^*)$.}
\Function{DRC}{$S'$}
    \State $v^*_{-}S' \gets \min\{E,\sum_{k\in S'}w_k\}$
    \If{$|S'| = 1$}
        \State \Return $v^*_{-}S'$
    \ElsIf{$\sum_{k\in S'}w_k\leq E$}
        \State \Return $(w_i)_{i\in S'}$
    \Else
        \For{each $j \in S'$}
            \State $\phi^{*[S'\setminus\{j\}]} \gets \Call{DRC}{S'\setminus \{j\}}$
        \EndFor
        \For{each $i \in S'$}
            \State $v^*_{-}S'\setminus\{i\} \gets \min\{E,\sum_{k\in S'\setminus \{i\}}w_k\}$
            \State $\displaystyle
            \phi^{*[S']}_i \gets
            \frac{1}{|S'|}
            \left(
                v^*_{-}S' - v^*_{-}S'\setminus\{i\}
                + \sum_{j\in S'\setminus \{i\}}\phi^{*[S'\setminus\{j\}]}_i
            \right)$
        \EndFor
        \State \Return $(\phi^{*[S']}_i)_{i \in S'}$
    \EndIf
\EndFunction

\Statex

\State Compute $\phi^* \gets \Call{DRC}{N}$
\State \Return $\phi^*$
\end{algorithmic}
\end{algorithm}

By employing the memoisation technique, 
    the time complexity of
    Algorithm~\ref{algo-recD}
    is bounded by $O(n^2 {\cal W})$, 
    assuming sufficient memory resources, 
    where ${\cal W}
    =\# \{S' \subseteq N\mid 
     |S'|\geq 2, \sum_{i\in S'} w_i > E\}$.

Since $E=W/2$ clearly implies $\cal W= \cal L$ and $\cal W$ is non-increasing while $\cal L$ is non-decreasing in $E$, comparing $E$ with $W/2$ suffices to determine which of the two recursive algorithms yields the no larger time complexity bound.

\section{Monte Carlo Method 
}\label{section5}

In this section, we introduce a Monte Carlo method for computing the Shapley value of bankruptcy games,
    which constitutes 
    a fully polynomial-time randomized approximation scheme
    (FPRAS). 
Throughout this section, we assume that 
    $E, w_1, \ldots ,w_n$ 
    are positive rational numbers 
    satisfying $E< W=\sum_{i \in N} w_i.$
First, we introduce a simple preprocedure.
If a player $i\in N$ satisfies $E < w_i$, 
    then $w_i$ is replaced by $\min\{w_i, E\}$, 
    under which the characteristic function $v$ remains unchanged.
This transformation reduces 
    the input size of the game, 
    and preserves (or potentially enhances) 
    computational efficiency.
In this section, we assume the following.
\begin{Assumption} \label{Assumption:wi<E}
    The claim $w_i$ of each player $i\in N$ satisfies
        $0< w_i \leq E.$ 
\end{Assumption}


We now describe a Monte Carlo method, 
    which corresponds to the most simple form 
    of the algorithms 
    for calculating the Shapley-Shubik index 
    proposed by Mann and Shapley in their seminal work~\cite{Mann1960RM2651}.
For an overview of Monte Carlo methods 
    in Shapley value (and/or Shapley-Shubik index)
    computation, 
    see~\cite{bachrach2010approximating,
    Liben-Nowell2012Computing,
    Mann1960RM2651,
    ushioda2022monte} for example.

\begin{algorithm}
\caption{}\label{algo1}
\begin{algorithmic}[Algorithm 1]
\Input A positive rational vector $(E; w_1, w_2, \dots, w_n)$
    satisfying 
    $E < \sum_{i=1}^nw_i$ 
    and a positive integer $M$.
\Output A payoff vector 
        $(\phi^A_1, \phi^A_2, \ldots , \phi^A_n)$.
\State Update $w_i \gets \min \{w_i, E\}$ $(\forall i \in N)$.
    \quad \Comment{Assumption~\ref{Assumption:wi<E}}
\State Set $m \gets 0$ and  $\phi'_i \gets 0$ $(\forall i \in N)$.
\While{$m < M$}
\State  Choose $\pi \in \Pi_N$ randomly.
\State \label{step1.2} Update $\phi'_i \gets \phi'_i + \Merg (\pi ,i )$ $(\forall i \in N)$ and $m \gets m + 1$.
\EndWhile
\State Output $\phi^A_i = \frac{1}{M}\phi'_i$ for each $i \in N$.
\end{algorithmic}
\end{algorithm}

%

\noindent 
The time complexity of Algorithm~\ref{algo1} is $O(Mn)$.

We denote the vector (of random variables) 
	obtained by executing Algorithm~\ref{algo1}
	by 	
	$(\phi^A_1, \phi^A_2,\ldots ,
		\phi^A_n)$.
The following results are obtained  straightforwardly.

\begin{Theorem}\label{Theorem:FPRAS}
For each player $i \in N$, 
	$\Exp \left[ \phi^A_i \right]=\phi_i$.

For any $\varepsilon >0$ and  $0< \delta <1$,
	we have the following.

\noindent 
{\rm (1)}
If we set 
	$\displaystyle M \geq \frac{n^2 \ln (2/\delta)}{2\varepsilon^2}$,
	then each player $i \in N$ satisfies that
\[
\Prob 
	\left[ 
		\frac{| \phi^A_i-\phi_i| }{\phi_i}
            < \varepsilon  
	\right] \geq 1-\delta.
\]

\noindent 
{\rm (2)}
 If we set 	
	$\displaystyle M\geq \frac{n^2 \ln (2n/\delta)}{2\varepsilon^2}$,
	then 
\[
 \Prob 
	\left[
		\forall i \in N, 
             \frac{| \phi^A_i-\phi_i| }{\phi_i}
            < \varepsilon  
	\right] \geq 1-\delta.
\]
\end{Theorem}

\begin{proof}
For each $m \in \{1, 2, \ldots, M\}$ and $i \in N$, 
    define the random variable 
    $\Delta^{(m)}_i = \Merg(\pi, i)$, 
    where $\pi \in \Pi_N$ is 
    the permutation sampled 
    in the $m$-th iteration of Algorithm~\ref{algo1}.
As Algorithm~\ref{algo1} chooses a permutation $\pi \in \Pi_N$ randomly, 
    it is obvious that for each player $i \in N$, 
	$\{ \Delta^{(1)}_i, \Delta^{(2)}_i, \ldots , \Delta^{(M)}_i\}$ 
	is a Bernoulli process satisfying
	$\phi^A_i=\sum_{m=1}^M \Delta^{(m)}_i/M$ and 
 	$ \Exp \left[ \phi^A_i \right]
	=\Exp \left[ \Delta^{(m)}_i \right]=\phi_i$ 
	$(\forall m \in \{1,2,\ldots ,M\})$.

The definition of Algorithm~\ref{algo1}
    directly implies that 
    $\Delta^{(m)}_i \in [0, w_i]$ $(\forall i \in N).$
In addition, 
    Assumption~\ref{Assumption:wi<E}  ensures that  
    $\Merg (\pi, i) = w_i$ holds
    for any $(\pi, i) \in \Pi_N \times N$
    such that $\pi (n)=i.$
Thus, we obtain a lower bound for 
    the Shapley value, satisfying 
\[
 \phi_i
 = \frac{1}{n!} \sum_{\pi \in \Pi_N} \Merg (\pi ,i)
 \geq \frac{1}{n!} \sum_{\pi \in \Pi_N: \pi (n)=i} 
    \Merg (\pi ,i)
 =\frac{(n-1)!}{n!} w_i 
 =\frac{w_i}{n} (>0)  \;\;\;
 (\forall i \in N).
\]

Hoeffding's inequality~\cite{hoeffding1963probability}
	implies that 
	each player $i \in N$ satisfies
\begin{eqnarray*}
\Prob \left[ \left| \phi^A_i
	- \Exp \left[ \phi^A_i \right] \right|
	\geq t \right]
&\leq 
& 2 \exp \left(
	- \frac{2M^2 t^2}{\sum_{m=1}^M ({w_i}-0)^2}
	\right)
= 2 \exp \left( \frac{-2M t^2}{{w_i}^2} \right).
\end{eqnarray*}

(1) If we set 
	$\displaystyle M \geq \frac{n^2 \ln (2/\delta)}{2\varepsilon^2}$,
	then 
\[
\begin{array}{l}
\displaystyle 
\Prob 
	\left[ 
		\frac{| \phi^A_i-\phi_i| }{\phi_i}
            < \varepsilon  
	\right] 
=
\Prob 
	\left[ 
		\left| \phi^A_i-\Exp \left[ \phi^A_i \right] \right| 
            < \varepsilon  \phi_i
	\right] 
\geq \Prob 
	\left[ 
		\left| \phi^A_i-\Exp \left[ \phi^A_i \right] \right| 
            < \varepsilon  \left( \frac{{w_i}}{n} \right) 
	\right]  \\
\displaystyle 
= 1- \Prob 
	\left[ 
		\left| \phi^A_i-\Exp \left[ \phi^A_i \right] \right| 
            \geq \left( \frac{\varepsilon {w_i}}{n} \right)
	\right]  
\geq 1- 
2 \exp \left( \frac{-2M (\varepsilon {w_i}/n)^2}{{w_i}^2} \right)\\
\displaystyle 
=1- 
2 \exp \left( -M \frac{2  \varepsilon^2}{n^2} \right)
\geq  
1- 2 \exp \left( - \ln \left( \frac{2}{\delta} \right) \right) =
  1-\delta.
 \end{array}
\] 

(2) If we set 
	$\displaystyle M\geq \frac{n^2 \ln (2n/\delta)}{2\varepsilon^2}$,
	then we have that
\[
\begin{array}{l}
\displaystyle 
\Prob \left[
				\forall i \in N, 
					\frac{\left| 
						\phi^A_i	- \varphi_i 
					\right|}{\phi_i}
					< \varepsilon 
			\right]
 = 1- \Prob 
			\left[
				\exists i \in N, 
					\left|
                        \phi^A_i-\Exp \left[ \phi^A_i \right] 
                    \right| 
					\geq \varepsilon \phi_i
			\right] \\
\displaystyle
 \geq 
 1-\sum_{i \in N} 
	\Prob 
			\left[
					\left|
                        \phi^A_i-\Exp \left[ \phi^A_i \right] 
                    \right| 
					\geq \varepsilon \phi_i
			\right]
\geq 
1-\sum_{i \in N} 
	\Prob 
			\left[
					\left|
                        \phi^A_i-\Exp \left[ \phi^A_i \right] 
                    \right| 
					\geq \varepsilon \left( \frac{{w_i}}{n} \right)
			\right] \\
\displaystyle 
\geq 1- \sum_{i=1}^n  2 \exp 
        \left( 
            \frac{-2M (\varepsilon {w_i} /n)^2 }{{w_i}^2}
        \right) 
\geq 1- \sum_{i=1}^n  2 \exp 
        \left( -M 
            \frac{2 \varepsilon^2 }{n^2}
        \right) \\
\displaystyle
 \geq 
 1- \sum_{i=1}^n  2 \exp 
	\left( 
		- \ln \left( \frac{2n}{\delta} \right) 
 	\right) 
= 1-\sum_{i=1}^n \frac{\delta}{n}=1-\delta.  
\end{array}
\]
\end{proof}

The algorithms proposed 
    in~\cite{saavedra2020systems,liben2012computing} 
    are essentially equivalent to Algorithm~\ref{algo1} 
    in structure and purpose. 
The analysis described in~\cite{saavedra2020systems}
    yields a bound of the time complexity 
    that depends on the numerical value
    of the positive integer weight~$w_i$
    (rather than on 
    the input size~$1+\lfloor \log_2 w_i \rfloor$),
    and therefore does not guarantee
   an FPRAS.
Liben-Nowell et al.~\cite{liben2012computing}
    showed that their algorithm 
    gives an FPRAS for general convex games
    using Chebyshev's inequality.

When $E \geq \tfrac{1}{2}W$, 
    fewer samples suffice.

\begin{Cor}\label{Theorem:MonteHalf}
For $E \geq \frac{1}{2}W$ 
    and any $\varepsilon>0,$ $0<\delta<1$, 
the following holds.

\noindent 
{\rm (1)}
If we set 
	$\displaystyle M \geq \frac{2 \ln (2/\delta)}{\varepsilon^2}$,
	then each player $i \in N$ satisfies that
\[
\Prob 
	\left[ 
		\frac{| \phi^A_i-\phi_i| }{\phi_i}
            < \varepsilon  
	\right] \geq 1-\delta.
\]

\noindent 
{\rm (2)}
 If we set 	
	$\displaystyle M\geq \frac{2 \ln (2n/\delta)}{\varepsilon^2}$,
	then 
\[
 \Prob 
	\left[
		\forall i \in N, 
             \frac{| \phi^A_i-\phi_i| }{\phi_i}
            < \varepsilon  
	\right] \geq 1-\delta.
\]
\end{Cor}

\begin{proof}
Lemmas~\ref{Theorem:w_0Monotone} 
and~\ref{Theorem:PropDist} imply that 
    $\phi_i \geq \tfrac{1}{2}w_i$ $(\forall i \in N)$.  
Applying this bound and replicating the argument 
    used in the proof of Theorem~\ref{Theorem:FPRAS} 
    yields the desired result.
\end{proof}

\section{Conclusion}\label{section6}

In this paper, we focus on the computational aspects of the Shapley value in bankruptcy games and make the following three contributions:
\begin{enumerate}
\item 
    We show that a decision problem 
    of computing the Shapley value 
    in bankruptcy games is NP-complete. 
    This finding suggests that 
    it is difficult to construct 
    a polynomial-time algorithm 
    for the exact computation of the Shapley value.
\item 
    We investigate the relationship 
    between the Shapley value of bankruptcy games 
    and the Shapley-Shubik index 
    in weighted voting games.
    Our result naturally leads to a dynamic programming technique 
    for calculating the Shapley value of a bankruptcy game 
    whose time complexity is bounded by $O(n^2 E)$ for each player, 
    assuming that $E,w_1,\ldots w_n$ are positive integers. 
    Our result improves the time complexity 
    of the algorithm proposed in~\cite{aziz2013computation}.
\item 
    We present two recursive algorithms 
    for computing the Shapley value in bankruptcy games.
    Unlike the dynamic programming approach, 
    these algorithms do not assume that 
    $E$ and $w_1,\ldots ,w_n$ are integers.
    We first describe the recursive completion method, 
    originally proposed by O'Neill.
    We also introduce our own recursive algorithm 
    based on the dual game formulation.
    We show that both approaches benefit from memoisation, 
    improving computational efficiency.
    Their time complexities are bounded 
    by $O(n^2 {\cal L})$ and $O(n^2 {\cal W})$, 
    where 
    ${\cal L}
    =\# \{S'' \subseteq N \mid 
        |S''| \leq n-2, \sum_{k\in S''} w_k < E\}$ and   
    ${\cal W}
    =\# \{S' \subseteq N\mid 
     |S'|\geq 2, \sum_{i\in S'} w_i > E\}$, respectively.
\item 
    We describe a Fully Polynomial-Time Randomized Approximation Scheme 
    (FPRAS) for computing the Shapley value.
    We have also shown that when $E \geq \tfrac{1}{2}W$, 
    the required number of samples
    grows sublinearly with respect to $n$
    (number of players).  
\end{enumerate}

\subsection*{Future Work}
An interesting direction for future work is to develop an FPRAS for computing the Shapley value in bankruptcy games that also satisfies monotonicity, 
i.e., ensuring that $w_i \leq w_j$ implies $\phi_i^A \leq \phi_j^A$, which the Monte Carlo method described in Section~\ref{section5} does not always guarantee.
This line of research is motivated by the construction 
    of an FPRAS for the Shapley-Shubik index in weighted voting games by Ushioda et al.~\cite{ushioda2022monte}.

\bibliography{refer}

\end{document}